\documentclass[doublecol,floats]{epl2}

\usepackage{graphicx,psfrag,amsmath,amssymb,amsfonts,bbm,latexsym,color,dcolumn,bm,mathbbol,ulem}

\definecolor{myred}{RGB}{168,5,14}

\newtheorem{theorem}{Theorem}
\newenvironment{proof}[1][Proof]{\noindent\textbf{#1.} }{\ \rule{0.5em}{0.5em}}

\title{Asymptotic lower bound for the gap of Hermitian matrices having
  ergodic ground states and infinitesimal off-diagonal elements}
\shorttitle{Lower bound for the gap of Hermitian matrices with
  infinitesimal off-diagonal elements}

\author{M. Ostilli \inst{1,2} and C. Presilla \inst{2,3}}
\institute{ 
\inst{1}
Departamento de Fisica, Universidade Federal de Santa Catarina, Florianopolis, 88040-900, SC, Brazil \\
\inst{2}
Dipartimento di Fisica, Universit\`a di Roma ‘La Sapienza’, Piazzale Aldo Moro 2, Roma I-00185, Italy\\
\inst{3}
Istituto Nazionale di Fisica Nucleare, Sezione di Roma 1,
  Roma 00185, Italy
} 


\abstract{
  Given a $M\times M$ Hermitian matrix $\bm{\mathcal{H}}$ with
  possibly degenerate eigenvalues
  $\mathcal{E}_1<\mathcal{E}_2<\mathcal{E}_3< \dots$, we provide, in
  the limit $M\to\infty$, a lower bound for the gap
  $\mu_2=\mathcal{E}_2-\mathcal{E}_1$ assuming that (i) the
  eigenvector (eigenvectors) associated to $\mathcal{E}_1$ is ergodic
  (are all ergodic) and (ii) the off-diagonal terms of
  $\bm{\mathcal{H}}$ vanish for $M\to\infty$ more slowly than
  $M^{-2}$.  Under these hypotheses, we find
  $\varliminf_{M\to\infty}\mu_2\geq
  \varlimsup_{M\to\infty}\min_{n}\mathcal{H}_{n,n}$.  This general
  result turns out to be important for upper bounding the relaxation
  time of linear master equations characterized by a matrix equal, or
  isospectral, to $\bm{\mathcal{H}}$.  As an application, we consider
  symmetric random walks with infinitesimal jump rates and show that
  the relaxation time is upper bounded by the configurations (or
  nodes) with minimal degree.
}

\begin{document}


\maketitle

\section{Introduction}
In classical and quantum physics, as well as in applied sciences, many
systems can be effectively described by linear master
equations~\cite{Glauber, Gardiner, Kampen, WeissG, Petruccione,
  Weiss}. If the system is characterized by $M$ states that we can
label with an index $n\in\{1,\dots,M\}$, the master equation reads
\begin{align}
  \label{ME}
  &\frac{d \bm{p}(t)}{dt}=-\bm{\mathcal{L}}~\bm{p}(t),
\end{align}
where $\bm{p}(t)^{T}=\left(p_1(t),\dots,p_n(t)\right)$ is a row vector
in which each component $p_n(t)$ represents the probability that the
system is in the state $n$ at time $t$. The matrix $\bm{\mathcal{L}}$
(a weighted Laplacian) is a singular $M\times M$ real matrix having the
property $\sum_{m}\mathcal{L}_{m,n}=0$, $n\in\{1,\ldots,M\}$, which
stems from the conservation of the total probability $\sum_m
p_m(t)=1$, and where, for $m\neq n$, $-\mathcal{L}_{m,n}\geq 0$
represents the transition rate from state $n$ to state $m$.

In some cases $\bm{\mathcal{L}}$ is symmetric and, therefore, has real
distinct eigenvalues $\mu_1=0,\mu_2,\mu_3,\ldots$. Furthermore, if the
associated stochastic matrix
$\bm{\mathcal{S}}=\bm{1}-r^{-1}\bm{\mathcal{L}}$, where
$r=\max_n\mathcal{L}_{n,n}$, is irreducible, the Perron-Frobenious
theorem \cite{Weyl} implies that $\mu_1=0$ is the minimal eigenvalue,
it is simple (because $\mu_1$ is real), and $\mu_k>0$ for $k\geq 2$.
Finally, we have that $p^{(\mathrm{eq})}_m = 1/M$ is the $m$-th
component of the normalized eigenvector $\bm{p}^{(\mathrm{eq})}$
corresponding to the eigenvalue $\mu_1=0$.

If $\bm{\mathcal{L}}$ is asymmetric, we expect, in general, complex
eigenvalues.  However, if we assume that a detailed balance condition
holds, i.e. there exist $M$ positive values, $p^{(\mathrm{eq})}_m>0$,
such that $\mathcal{L}_{m,n} p^{(\mathrm{eq})}_n = \mathcal{L}_{n,m}
p^{(\mathrm{eq})}_m$, then $\bm{\mathcal{L}}\bm{p}^{(\mathrm{eq})}=0$ ($\bm{p}^{(\mathrm{eq})}$ is a right eigenvector), and
the spectrum of $\bm{\mathcal{L}}$ is still
real and non negative. In fact, $\bm{\mathcal{L}}$ is similar (and therefore isospectral) to the
real symmetric matrix $\bm{\mathcal{L}}_s = \bm{\mathcal{R}}^{-1}
\bm{\mathcal{L}} \bm{\mathcal{R}}$, where $\bm{\mathcal{R}}$ is the
diagonal matrix defined by the elements $\mathcal{R}_{m,n} =\delta_{m,n} ( p^{(\mathrm{eq})}_m )^{1/2}$.
Note that the condition $p^{(\mathrm{eq})}_m>0$ for any $m$, which ensures
the existence of the inverse $\bm{\mathcal{R}}^{-1}$, guarantees that
both $\mathcal{L}_{m,n}$ and $\mathcal{L}_{n,m}$ are 0 if one of the
two is so.

Whether $\bm{\mathcal{L}}$ is symmetric or not, but satisfies a
detailed balance condition and has an associated irreducible
stochastic matrix $\bm{\mathcal{S}}$, its spectrum consists of real
eigenvalues $0=\mu_1 <\mu_2 < \mu_3 < \dots$.  The eigenvalue
$\mu_1=0$ is simple and the corresponding eigenvector
$\bm{p}^{(\mathrm{eq})}$ is necessarily an ergodic ground state,
i.e. all its components are positive, and represents the unique
stationary state of Eq.~(\ref{ME}) toward which the system eventually
converges. For almost all initial conditions, up to terms
exponentially smaller, we have $\|\bm{p}(t)-\bm{p}^{(\mathrm{eq})} \|
\sim C \exp(-\mu_2 t)$, where $C$ is a constant.  In other words,
$\mu_2$, the minimal non zero eigenvalue of $\bm{\mathcal{L}}$,
provides the inverse of the relaxation time to
equilibrium. Determining $\mu_2$ is evidently of crucial
importance. In particular, the existence of a finite lower bound to
$\mu_2$ in the limit of $M\to\infty$ is pivotal in establishing if
$\bm{\mathcal{L}}$ represents a gapped system~\cite{CP-GW}.

Motivated by the above remarks, we consider a generic real symmetric or, more in general, 
Hermitian matrix $\bm{\mathcal{H}}$ and let $\mathcal{E}_1 < \mathcal{E}_2<
\mathcal{E}_3< \dots$ be its distinct, possibly degenerate,
eigenvalues.  We are interested in evaluating
$\mu_2=\mathcal{E}_2-\mathcal{E}_1$.  The matrix $\bm{\mathcal{H}}$
can be seen as a $M$-dimensional Hamiltonian operator.  In fact, we
can always split $\bm{\mathcal{H}}$ as
\begin{align}
  \label{LB}
  \bm{\mathcal{H}}=\bm{\mathcal{V}}+\bm{\mathcal{K}},
\end{align}
where, on the given canonical base, or base of the configurations,
$n\in\{1,\ldots,M\}$,
\begin{align}
  \label{LB1}
  {\mathcal{V}}_{m,n}=\delta_{m,n}{\mathcal{V}}_{n},
\end{align}
and
\begin{align}
  \label{LB2}
  {\mathcal{K}}_{m,n}=-\left(1-\delta_{m,n}\right)\sigma_{m,n},
\end{align}
the definitions of the vector ${\mathcal{V}}_{n}$ and matrix
$\sigma_{m,n}$ being implicit. Note that $\sigma_{m,n}$ is an Hermitian matrix with
$\sigma_{m,m}=0$ and its off-diagonal elements have non-defined signs or phases.
With such a
decomposition, $\bm{\mathcal{V}}$ and $\bm{\mathcal{K}}$ play the role
of ``potential'' and ``kinetic'' operators, respectively.  Suppose
that the lowest eigenstate $\mathcal{E}_1$ is $k$-fold degenerate,
with $k$ finite, and let
$|\mathcal{E}_1^{(1)}\rangle,\ldots,|\mathcal{E}_1^{(k)}\rangle$ be
the corresponding orthonormalized eigenstates.  We represent the
components of each ground state (GS) as
\begin{align}
  \label{LB3}
  \langle n|\mathcal{E}_1^{(i)}\rangle =\frac{u^{(i)}_n}{\sqrt{M}},
  \qquad i=1,\ldots,k.
\end{align}
In case $\mathcal{E}_1$ is simple, i.e. $k=1$, we will omit the
superscript $^{(i)}$. 

In this paper, we state and prove a general lower bound for the gap
$\mu_2=\mathcal{E}_2-\mathcal{E}_1$ valid for a large class of 
Hermitian matrices $\bm{\mathcal{H}}$ whose unique or multiple ground
states are all ergodic.  By an ergodic GS here we mean that, for any
$M$, each component of the GS is non zero, and
finite (apart from the normalization condition).  More precisely, we
say that the GS $|\mathcal{E}_1\rangle$ is ergodic if 
\begin{align}
  \label{LB4}
  0< \varliminf_{M\to\infty} |u_n|,\quad \forall n.
\end{align}
In other words, $|\mathcal{E}_1\rangle$ is ergodic if, as a function
of $M$
\begin{align}
  \label{LB5}
  \langle n|\mathcal{E}_1\rangle
  =\mathop{O}\left(\frac{1}{\sqrt{M}}\right),\quad \forall n.
\end{align}

In the next Section, we precisely state and prove the lower bound in the form
of a theorem valid for arbitrary finite degeneracy of $\mathcal{E}_1$.
Note that the index $k$ of this degeneracy is thought to be a constant independent of
the size $M$. Then we apply the result to symmetric random
walks characterized by infinitesimal jump rates, and show that the
relaxation time $\tau$ is upper bounded by the minimal degree of the
configurations (called nodes in graph theory).

\section{Lower bound for $\mu_2$}
We shall make use of the definitions (\ref{LB}-\ref{LB4}) previously
introduced.
\begin{theorem}
  (Non degenerate case) Let $\bm{\mathcal{H}}$ be a $M\times M$ 
  Hermitian matrix with distinct, possibly degenerate, eigenvalues
  $\mathcal{E}_1<\mathcal{E}_2<\mathcal{E}_3<\ldots$.  Let us suppose
  that the GS of $\bm{\mathcal{H}}$ is unique 
  and that there exists a
  positive function $g(M)$ such that 
  \begin{align}
    \label{T00}
    \sigma=\max_{m,n}\frac{|\sigma_{m,n}|}{|u_nu_m|}<g(M),
  \end{align}
  \begin{align}
    \label{T01}
    \lim_{M\to\infty} g(M)=0, \quad \lim_{M\to\infty}
    \frac{1}{M^2g(M)}=0.
  \end{align}
  Let the GS of $\bm{\mathcal{H}}$ be ergodic, and
  $\varlimsup_{M\to\infty} \mathcal{E}_1/M=0$.  Then, for the gap
  $\mu_2=\mathcal{E}_2-\mathcal{E}_1$, we have
  \begin{align}
    \label{T1}
    \varliminf_{M\to\infty} \mu_2 \geq \varlimsup_{M\to\infty}
    \min_{n} \mathcal{V}_n.
  \end{align}
\end{theorem}

The same result holds essentially unchanged if $\mathcal{E}_1$ is
$k$-fold degenerate, with $k$ finite and independent of $M$, and each
one corresponding GS is ergodic.
\begin{theorem}
  (Finite degenerate case) Let $\bm{\mathcal{H}}$ be a $M\times M$
  Hermitian matrix with distinct, possibly degenerate,
  eigenvalues $\mathcal{E}_1<\mathcal{E}_2<\mathcal{E}_3<\ldots$.  Let
  us suppose that we have $k$ degenerate GSs 
  and there exists a positive function $g(M)$ such that
  \begin{align}
    \label{T00g}
    \sigma=\max_i
    \max_{m,n}\frac{|\sigma_{m,n}|}{|u^{(i)}_nu^{(i)}_m|}<g(M),
  \end{align}
  \begin{align}
    \label{T01g}
    \lim_{M\to\infty} g(M)=0, \quad \lim_{M\to\infty}
    \frac{1}{M^2g(M)}=0.
  \end{align}
  Let each GS of $\bm{\mathcal{H}}$ be ergodic, and
  $\varlimsup_{M\to\infty} \mathcal{E}_1/M=0$.  Then, for the gap
  $\mu_2=\mathcal{E}_2-\mathcal{E}_1$, we have
  \begin{align}
    \label{T1g}
    \varliminf_{M\to\infty} \mu_2 \geq \varlimsup_{M\to\infty}
    \min_{n} \mathcal{V}_n.
  \end{align}
\end{theorem}

\begin{proof}
  We shall make use of the following short-hand notation: given a
  Hermitian operator $\bm{C}$, $GSL[\bm{C}]$ stands for the GS level
  of $\bm{C}$ (i.e., the minimal eigenvalue of $\bm{C}$).

  Let us first suppose that the GS of $\bm{\mathcal{H}}$ is unique and
  that $\mathcal{E}_1=0$ (so that $\mu_2=\mathcal{E}_2$).  Let us
  introduce the following new Hamiltonian
  \begin{align}
    \label{T2}
    \bm{\mathcal{F}}(\lambda)=
    \bm{\mathcal{H}}+\lambda|\mathcal{E}_1\rangle\langle\mathcal{E}_1|.
  \end{align}
  We have
  \begin{align}
    \label{T4}
    GSL\left[\bm{\mathcal{F}}(\lambda)\right]=\left\{
      \begin{array}{ll}
        \lambda, &\lambda<\mu_2,\\
        \mu_2 &\lambda\geq \mu_2.
      \end{array}
    \right.
  \end{align}
  Equation~(\ref{T4}) in particular implies that (the limit exists due
  to the monotonicity)
  \begin{align}
    \label{T5}
    \mu_2=\lim_{\lambda\to\infty}GSL\left[\bm{\mathcal{F}}(\lambda)\right].
  \end{align}

  \textit{Uniform ergodic state.}  For the moment being, let us
  suppose that $\langle n|\mathcal{E}_1\rangle=1/\sqrt{M}$ (such a
  situation occurs, for example, when we are considering a random
  walk, where $\mathcal{V}_n=-\sum_{m}\mathcal{K}_{m,n}$).  In this
  case, we can rewrite Eq. (\ref{T2}) as
  \begin{align}
    \label{T2b}
    \bm{\mathcal{F}}(\lambda)=
    \bm{\mathcal{V}}+\frac{\lambda}{M}\bm{I}+\bm{\mathcal{K}}+
    \frac{\lambda}{M}\bm{\mathcal{U}},
  \end{align}
  where we have introduced
  \begin{align}
    \label{U}
    \bm{\mathcal{U}}=M|\mathcal{E}_1\rangle\langle\mathcal{E}_1|-\bm{I},
  \end{align}
  and $\bm{I}$ is the identity matrix.  Note that, in the canonical
  base, we have $\mathcal{U}_{n,n}=0$ and $\mathcal{U}_{m,n}=1$ for
  $m\neq n$.  Furthermore, from
  \begin{subequations}
    \begin{align}
      \label{U1}
      &\bm{\mathcal{U}}|\mathcal{E}_1\rangle =
      (M-1)|\mathcal{E}_1\rangle,
      \\
      &\bm{\mathcal{U}}|\mathcal{E}_i\rangle = -|\mathcal{E}_i\rangle,
      \qquad i \neq 1,
    \end{align}
  \end{subequations}
  it follows that for any $\alpha\in \mathbb{R}$
  \begin{align}
    \label{U2}
    GSL\left[\alpha \bm{\mathcal{U}}\right]=\left\{
      \begin{array}{ll}
        -|\alpha|(M-1), &\alpha< 0, \\
        -|\alpha|, &\alpha\geq 0. 
      \end{array}
    \right.
  \end{align}
  Equation~(\ref{U2}), despite its simplicity, is the key of our
  proof: when $\alpha$ changes from a negative to a positive value,
  the GS level of $\alpha\bm{\mathcal{U}}$ changes from being
  extensive, i.e. of order $\mathop{O}(M)$, to being intensive,
  i.e. of order $\mathop{O}(1)$.  From the first Weyl's inequality
  \cite{Weyl}, we have
  \begin{align}
    \label{TW}
    GSL\left[\bm{\mathcal{F}}(\lambda)\right]\geq \min_n
    \mathcal{V}_n+\frac{\lambda}{M}+f\left(\frac{\lambda}{M},\bm{\sigma}\right),
  \end{align}
  where
  \begin{align}
    \label{T6}
    f\left(\frac{\lambda}{M},\bm{\sigma}\right)=GSL\left[\bm{\mathcal{K}}
      +\frac{\lambda}{M}\bm{\mathcal{U}}\right].
  \end{align}
  We are not able to exactly calculate
  $f\left(\frac{\lambda}{M},\bm{\sigma}\right)$, however, from
  Eq.~(\ref{U2}) we see that
  \begin{align}
    \label{T7}
    f\left(\frac{\lambda}{M},\bm{\sigma}\right)\sim \left\{
      \begin{array}{ll}
        -(M-1)\frac{\lambda}{M}+GSL\left[\bm{\mathcal{K}}\right] 
        &\frac{\lambda}{M} \ll \sigma, \\
        -\left(\frac{\lambda}{M}-\sigma^*\right), 
        &\frac{\lambda}{M} \gg \sigma, 
      \end{array}
    \right.
  \end{align}
  where $\sigma$ has been defined in Eq.~(\ref{T00}), and $\sigma^*\in\mathbb{R}$
  is some appropriate value of the order of magnitude of
  $2\sum_{m,n}\sigma_{m,n}/(M(M-1))$ (which is real).  We stress that we do not need
  to know $\sigma^*$, nor to assume Eq.~(\ref{T7}) as an actual
  equality.  However, Eq.~(\ref{T7}) makes clear that there exists a
  threshold in $\lambda/M$, which is of order $\sigma$, where there
  occurs a sort of phase transition, the GS of
  $\bm{\mathcal{K}}+\bm{\mathcal{U}}{\lambda}/{M}$ transiting from
  being extensive to being intensive. In the latter phase, we see that
  there exists a regime where $\lambda/M$, $\sigma$ and $\sigma^*$
  tend all to zero.  In fact, let us choose $\lambda=\lambda(M)$ such
  that
  \begin{align}
    \label{T10}
    \frac{\lambda(M)}{M}=\sqrt{g(M)}.
  \end{align}
  With this choice and due to Eqs.~(\ref{T00}), we have that the
  following limits are simultaneously satisfied
  \begin{align}
    \label{T11d}
    & \lim_{M\to\infty}\sigma=\lim_{M\to\infty}\sigma^*=0,\\
    \label{T11a}
    & \lim_{M\to\infty}\frac{\lambda(M)}{M}=0,\\
    \label{T11b}
    & \lim_{M\to\infty}\frac{\sigma}{\frac{\lambda(M)}{M}}=0,\\
    \label{T11c}
    & \lim_{M\to\infty}\lambda(M)=+\infty.
  \end{align}
  Equations~(\ref{T11d}-\ref{T11b}) plugged into Eq.~(\ref{T6})
  provide
  \begin{align}
    \label{T6}
    \lim_{M\to\infty}
    f\left(\frac{\lambda(M)}{M},\bm{\sigma}\right)=0,
  \end{align}
  whereas, by using Eq.~(\ref{T11c}) in Eq.~(\ref{TW}) we have
  \begin{align}
    \label{TWl}
    \varliminf_{M\to \infty}
    GSL\left[\bm{\mathcal{F}}(\lambda(M))\right]\geq \varlimsup_{M\to
      \infty} \min_n \mathcal{V}_n.
  \end{align}
  Finally, by using Eq.~(\ref{T5}) the proof of the theorem in the
  case of a uniform ergodic state is complete.

  \textit{General ergodic state.}  Now we still consider a unique (and
  ergodic) GS with $\mathcal{E}_1=0$, but we have $\langle
  n|\mathcal{E}_1\rangle=u_n/\sqrt{M}$ with $u_n\neq 1$.  Little
  changes are necessary to generalize the previous proof to the
  present case.  We define
  \begin{align}
    \label{T2bg}
    \bm{\mathcal{F}}(\lambda)=
    \bm{\mathcal{V}}+\frac{\lambda}{M}\bm{\mathcal{D}}
    +\bm{\mathcal{K}}+\frac{\lambda}{M}\bm{\mathcal{U}},
  \end{align}
  where
  \begin{align}
    \label{Ug}
    \bm{\mathcal{U}}=
    M|\mathcal{E}_1\rangle\langle\mathcal{E}_1|-\bm{\mathcal{D}},
  \end{align}
  and $\bm{\mathcal{D}}$ is a diagonal matrix with elements
  \begin{align}
    \label{D}
    {\mathcal{D}}_{m,n}=|u_n|^2\delta_{m,n}.
  \end{align}
  In the canonical base we have $\mathcal{U}_{n,n}=0$ and
  $\mathcal{U}_{m,n}=\overline{u_m}u_n$ for $m\neq n$.  Furthermore, from
  \begin{subequations}
    \label{U1g}
    \begin{align}
      &\langle n|\bm{\mathcal{U}}|\mathcal{E}_1\rangle =
      (M-|u_n|^2)\langle n|\mathcal{E}_1\rangle,
      \\
      &\langle n|\bm{\mathcal{U}}|\mathcal{E}_i\rangle = -|u_n|^2\langle
      n|\mathcal{E}_i\rangle, \qquad i\neq 1,
    \end{align}
  \end{subequations}
follows that, for any $\alpha\in \mathbb{R}$, we have
  \begin{align}
    \label{U2g}
    GSL\left[\alpha\bm{\mathcal{U}}\right] \sim \left\{
      \begin{array}{ll}
        -|\alpha|(M-s^*(M)), &\alpha< 0, \\
        -|\alpha|s^*(M), &\alpha\geq 0,
      \end{array}
    \right.
  \end{align}
  where $s^*(M)$ is an appropriate value of the order of magnitude
  of $s(M)$: 
  \begin{eqnarray}
    \label{umax}
    s(M)=\max_n |u_n|^2.  
  \end{eqnarray}
Eq. (\ref{U2g}) can be verified rigorously by using the 'Matrix Determinant Lemma' \cite{Weyl}.
More precisely, $s^*$ is the smallest root of the Equation in $\mu$: $\sum_n |u_n|^2/(\mu+|u_n|^2)=1$.
Notice that due to the normalization of $|\mathcal{E}_1\rangle$ we have $\sum_n |u_n|^2=M$ which, combined with the ergodicity condition (\ref{LB4}) implies
  \begin{eqnarray}
    \label{umax1}
    \varlimsup_{M\to \infty} s(M)<\infty \quad \mathrm{and} \quad \varlimsup_{M\to \infty} s^*(M)<\infty.  
  \end{eqnarray}

From the first Weyl's inequality, we have
  \begin{align}
    \label{TWg}
    GSL\left[\bm{\mathcal{F}}(\lambda)\right]\geq \min_n
    \left(\mathcal{V}_n+\frac{\lambda}{M}|u_n|^2\right)
    +f\left(\frac{\lambda}{M},\bm{\sigma}\right),
  \end{align}
  where
  \begin{align}
    \label{T6g}
    f\left(\frac{\lambda}{M},\bm{\sigma}\right)=
    GSL\left[\bm{\mathcal{K}}+\frac{\lambda}{M}\bm{\mathcal{U}}\right].
  \end{align}

  From Eq. (\ref{U2g})we have
  \begin{align}
    \label{T7g}
    f\left(\frac{\lambda }{M},\bm{\sigma}\right)\sim \left\{
      \begin{array}{ll}
        -(M-1)\frac{\lambda s^*(M)}{M}+GSL\left[\bm{\mathcal{K}}\right] 
        &\frac{\lambda}{M} \ll \sigma, \\
        -\left(\frac{\lambda s^*(M)}{M}-\sigma^*\right), 
        &\frac{\lambda}{M} \gg \sigma, 
      \end{array}
    \right.
  \end{align}
  where $\sigma$ has been defined in Eq.~(\ref{T00}), and $\sigma^*\in\mathbb{R}$
  is some appropriate value of the order of magnitude of
  $2\sum_{m,n}\sigma_{m,n}/(M(M-1)|u_n||u_m|)$ (which is real).  As in the uniform case, we do not need
  to know $\sigma^*$, nor to assume Eq.~(\ref{T7g}) as an actual
  equality. Furthermore, from Eqs. (\ref{T00}), (\ref{T01}), and (\ref{umax1}), there 
  exists a regime where $\lambda(M)s^*(M)/M$, $\sigma$ and $\sigma^*$
  tend all to zero.  

  So far we have considered, for simplicity, $\mathcal{E}_1=0$.  The
  generalization to the case $\mathcal{E}_1 \neq 0$ with
  $\varlimsup_{M\to \infty} \mathcal{E}_1/M=0$ is immediate. We
  replace $\lambda$ with $\lambda+\mathcal{E}_1$ in Eq.~(\ref{T2}) or
  in Eq.~(\ref{T2bg}).  The result of Eq.~(\ref{T1}) still holds since
  $\varlimsup_{M\to\infty} \mathcal{E}_1/M=0$.

  \textit{Degenerate case.}  Let us consider a two-fold degenerate
  case. We introduce
  \begin{align}
    \label{T2gg}
    \bm{\mathcal{F}}(\lambda^{(1)},\lambda^{(2)}) =& \bm{\mathcal{H}}
    +\left(\lambda^{(1)}+\mathcal{E}_1\right)
    |\mathcal{E}^{(1)}_1\rangle\langle\mathcal{E}^{(1)}_1| \nonumber
    \\ &+ \left(\lambda^{(2)}+\mathcal{E}_1\right)
    |\mathcal{E}^{(2)}_1\rangle\langle\mathcal{E}^{(2)}_1|.
  \end{align}
  Instead of Eq.~\ref{T5}, we now have to exploit
  \begin{align}
    \label{T5g}
    \mu_2=\lim_{\lambda^{(1)}\to\infty,\lambda^{(2)}\to\infty}
    GSL\left[\bm{\mathcal{F}}(\lambda^{(1)},\lambda^{(2)})\right].
  \end{align}
  By defining
  \begin{align}
    \label{Ugg}
    \bm{\mathcal{U}}=
    M\left(|\mathcal{E}^{(1)}_1\rangle\langle\mathcal{E}^{(1)}_1|
      +|\mathcal{E}^{(2)}_1\rangle\langle\mathcal{E}^{(2)}_1|\right)
    -\bm{\mathcal{D}}^{(1)}-\bm{\mathcal{D}}^{(2)},
  \end{align}
  where $\bm{\mathcal{D}}^{(1)}$ and $\bm{\mathcal{D}}^{(2)}$ have
  matrix elements
  \begin{align}
    \label{Dgg}
    {\mathcal{D}}^{(i)}_{m,n}=\left(u^{(i)}_n\right)^2\delta_{m,n},
    \qquad i=1,2,
  \end{align}
  we proceed as in the non-degenerate case.  The generalization to a
  $k$-fold degeneracy, with $k$ finite and independent of $M$, is
  obvious.
\end{proof}

\section{Weak ergodicity}
From the comments following Eq. (\ref{T7g}), it is evident that, in order for the theorem to hold,
we can ask for a weaker condition on the GS's $|\mathcal{E}_1^{(i)}\rangle$.
The following theorem accounts for such a generalization.

\begin{theorem}
  (Weak ergodicity) Let $\bm{\mathcal{H}}$ be a $M\times M$
  Hermitian matrix with distinct, possibly degenerate,
  eigenvalues $\mathcal{E}_1<\mathcal{E}_2<\mathcal{E}_3<\ldots$.  Let
  us suppose to have $k$ degenerate GSs 
  and there exists a positive function $g(M)$ such that
  \begin{align}
    \label{T00g}
    \sigma=\max_i
    \max_{m,n}\frac{|\sigma_{m,n}|}{|u^{(i)}_nu^{(i)}_m|}<g(M),
  \end{align}
  \begin{align}
    \label{T01g}
    \lim_{M\to\infty} g(M)=0, \quad \lim_{M\to\infty}
    \frac{1}{M^2g(M)}=0.
  \end{align}
  Let each GS of $\bm{\mathcal{H}}$ be weakly ergodic, i.e. for any $i$:
\begin{align}
  \label{LB4w}
  0< \varliminf_{M\to\infty} |u^{(i)}_n| \quad
\mathrm{if~} \exists m: ~\sigma_{m,n}\neq 0 \quad \mathrm{or~} \sigma_{n,m}\neq 0,
\end{align}
and
\begin{align}
  \label{LB4wb}
\varliminf_{M\to\infty}s^*(M)\sqrt{g(M)}=0,
\end{align}
where $s^*$ is the smallest root of the Equation in $\mu$: $\sum_n |u^{(i)}_n|^2/(\mu+|u^{(i)}_n|^2)=1$.  
Then, if $\varlimsup_{M\to\infty} \mathcal{E}_1/M=0$, for the gap
  $\mu_2=\mathcal{E}_2-\mathcal{E}_1$, we have
  \begin{align}
    \label{T1g}
    \varliminf_{M\to\infty} \mu_2 \geq \varlimsup_{M\to\infty}
    \min_{n} \mathcal{V}_n.
  \end{align}
\end{theorem}

\section{Application to random walks with infinitesimal jump rates}
A master equation of the form of Eq.~(\ref{ME}) can be interpreted as
a weighted continuous-time random walk taking place on a graph whose
adjacency matrix $\bm{\mathcal{A}}$ is defined by
$\mathcal{A}_{m,n}=(1-\delta_{m,n})\theta[-\mathcal{L}_{m,n}]$, where
$\theta[\cdot]$ is the Heaviside function.  The evolution of the
probability $p_n(t)$ goes through random jumps characterized by the
jump rates
\begin{align}
  \label{WW}
  W(n\to m)=-\mathcal{L}_{m,n}(1-\delta_{m,n}).
\end{align}
In the case of unweighted random walks, the non zero off-diagonal
elements of $\bm{\mathcal{L}}$ are uniform, whereas the diagonal
elements $\mathcal{L}_{n,n}$ coincide with the degree $k(n)$ of the
configuration with label $n$, namely,
\begin{align}
  \label{degree}
  \mathcal{L}_{n,n}=k(n)=\sum_{m\neq n} W(n\to m) =-\sum_{m\neq
    n}\mathcal{L}_{m,n}.
\end{align}
Furthermore, if the random walk is symmetric, i.e. the induced graph
is indirect, we have $\mathcal{L}_{m,n}=\mathcal{L}_{n,m}$.  In such
a case, we can directly identify $\bm{\mathcal{L}}$ as the Hamiltonian
$\bm{\mathcal{H}}$.  The GS of $\bm{\mathcal{H}}=\bm{\mathcal{L}}$ has
zero energy, $\mathcal{E}_1=0$, and is uniform, $\langle
n|\mathcal{E}_1\rangle=1/\sqrt{M}$.

Consider now a symmetric random walk in which the jump rates are
infinitesimal with $M$. We assume, for instance,
\begin{align}
  |\mathcal{L}_{m,n}|\leq \frac{1}{\left[\log(M)\right]^\alpha},
  \qquad m\neq n, \qquad 0<\alpha<1.
\end{align}
Taking into account that the GS is ergodic, we see that for this model
the hypotheses of the theorem are satisfied, and we can conclude that
(we keep using the symbol $k(n)$ as a weighted degree to include also
symmetric weighted random-walks because, although these have non-uniform jump-rates, their GS is still the uniform one)
\begin{align}
  \label{RW}
  \varliminf_{M\to\infty}\mu_2 \geq \varlimsup_{M\to\infty} \min_n
  k(n).
\end{align}
Equation~(\ref{RW}) tells us that, in symmetric random walks having
jump rates decaying with the logarithm of the system size $M$, the
relaxation time $\tau$ to reach equilibrium is upper bounded by the
configuration having minimal degree, namely $\tau\leq 1/ \min_n k(n)$.
This result is quite intuitive: a necessary condition to observe a
fast dynamics is that the configurations (or nodes) have not too small
degrees.  However, we warn the reader that, in the usual definition of
random walk, the jump rates are either $W(n\to m)=0$ or $W(n\to
m)=\mathop{O}(1)$, so that the theorem cannot be directly applied, its
hypotheses being not satisfied.  Nevertheless, if $\min_n k(n)$ is a
function of $M$ which diverges as $1/g(M)$, with
$1/(M^2g(M))\xrightarrow{M\to\infty} 0$, we can apply the theorem to
the matrix $\bm{\mathcal{H}}=\bm{\mathcal{L}}/(\min_n k(n))$, and we
conclude that
\begin{align}
  \label{RW1}
  \varliminf_{M\to\infty}\frac{\mu_2}{\min_n k(n)} \geq 1.
\end{align}

\section{Conclusions}
We have stated and proved a lower bound, in the limit $M\to\infty$,
for the gap of Hermitian $M\times M$ matrices characterized by (i)
ergodic GS, simple or degenerate, and (ii) off-diagonal terms which
are infinitesimal in $M$.  The ``infinitesimal'' conditions
(\ref{T00}-\ref{T01}) under which the theorem is satisfied are quite
mild, and, in particular, cover the common cases in which the
off-diagonal terms decay as the logarithm of the system size $M$. This
includes random-walks models with infinitesimal jump rates,
thermalization of classical models characterized by infinitesimal
couplings arbitrarily distributed, and the Pauli master equation
characterizing the thermalization of open quantum systems
\cite{Petruccione, Weiss}.

The key ingredient of the proof of the present theorem is
Eq.~(\ref{U2}), or its generalization Eq.~(\ref{U2g}).  From a physics
viewpoint, these equations describe a sort of phase transition between
a ``bosonic'' extensive phase, and a ``non-bosonic'' intensive
phase. Here, the terms ``bosonic'' and ``non-bosonic'' come from the
observation of the signs of the off-diagonal matrix elements of the
operator $\alpha\bm{\mathcal{U}}$, which, in turn, determine those of
the ``kinetic'' operator $\bm{\mathcal{K}}+\bm{\mathcal{U}}\lambda/M$.
Essentially, the proof of the theorem consists in finding a scaling
$\lambda(M)$ such that the system stays in the non-bosonic intensive
phase, and, at the same time, the parameters characterizing
$\bm{\mathcal{K}}+\bm{\mathcal{U}}\lambda/M$ tend to zero.

To compare our theorem with previously known results, a few comments
are in order.  Consider a symmetric and unweighted Laplacian matrix
$\bm{\mathcal{L}}$, with $\mathcal{L}_{n,n}=k(M)$, where $k(M)$ is a
suitable growing function of $M$.  We can apply the theorem to the
matrix $\bm{\mathcal{H}}=\bm{\mathcal{L}}/k$ obtaining
$\varliminf_{M\to\infty}\mu_2\geq 1$. On the other hand, the graph
induced by $\bm{\mathcal{L}}$ is a regular graph of degree $k$, and
classical results on spectral graph theory applied to
$\bm{\mathcal{H}}$ provide the more accurate estimate $\mu_2\simeq
1-2/(\sqrt{k-1})$ \cite{BrouwerHaemers}.  This example shows that for
matrices associated to graphs satisfying special properties, our lower
bound can be somehow not competitive.  However, our result may become
crucial whenever, besides the above conditions (i) and (ii), there are
no other assumptions or information about $\bm{\mathcal{H}}$.  Notice,
in particular, that $\bm{\mathcal{H}}$ can be quite different from a
Laplacian, its matrix elements being weighted and with no definite
sign or phase.  We are not aware of other lower bounds for such a general case.
In Ref.~\cite{AC2005} general weighted non symmetric Laplacians
$\bm{\mathcal{L}}$ are considered and the eigenvalues of the
standardized Laplacian $\bm{\mathcal{L}}/M$ are bounded in a region of
the complex plane which contains the real segment $[0,1]$.  If we
assume a Laplacian with real eigenvalues and take
$\mathcal{L}_{n,n}=\alpha_n M$, with $0\leq \alpha_n\leq 1$, and
$\mathcal{L}_{m,n}=-\beta_{m,n}$, with $0\leq \beta_{m,n}\leq 1$,
$m\neq n$, for the standardized Laplacian our theorem provides
$\varliminf_{M\to\infty}\mu_2\geq \varlimsup_{M\to\infty} \min_{k}
\alpha_k$, asymptotically localizing, as a function of the values of
the set $\{\alpha_k\}$, the second eigenvalue in the same segment
$[0,1]$ of Ref.~\cite{AC2005}.  However, we stress that a standardized
Laplacian with its off-diagonal elements vanishing as $1/M$ is not of
great interest for physics, whereas our theorem covers the widespread
cases in which the off-diagonal terms decay as $1/\log M$, i.e.
linearly with the inverse of the system size $N$ where, typically, $N$
is the number of particles or the physical volume.

\begin{acknowledgments}
M. O. acknowledges CNPq Grant PDS 150934/2013-0 and Sapienza University of Rome grant \textit{Professori Visitatori 2014} prot. C26V14LS9J.
\end{acknowledgments}


\begin{thebibliography}{99}

\bibitem{Glauber} R.~J.~Glauber, J. Math. Phys. \textbf{4}, 294
  (1963).

\bibitem{Gardiner} C.~W.~Gardiner, \textit{Handbook of Stochastic
    Methods} (Springer, 1985).

\bibitem{Kampen} N.~G.~Van Kampen, \textit{Stochastic Processes in
    Physics and Chemistry} (North-Holland, 1992).

\bibitem{WeissG} G.~H.~Weiss, \textit{Aspects and Applications of the
    Random Walk, Random Materials and Processes} (North-Holland
  Publishing, 1994).

\bibitem{Petruccione} H.-P.~Breuer and F.~Petruccione, \textit{The
    Theory of Open Quantum Systems} (Oxford University Press, 2002).

\bibitem{Weiss} U.~Weiss, \textit{Quantum Dissipative Systems} (World
  Scientific, 2008).

\bibitem{Weyl} J.~N.~Franklin, \textit{Matrix Theory} (Dover
  Publications, 1993).

\bibitem{CP-GW} T.~S.~Cubitt, D.~Perez-Garcia3, and M.~M.~Wolf, Nature
  \textbf{528}, 207 (2015).

\bibitem{BrouwerHaemers} A.~E.~Brouwer and W.~H.~Haemers,
  \textit{Spectra of graphs} (Springer, 2012).

\bibitem{AC2005} R.~Agaev and P.~Chebotarev, Linear Algebra
  Appl. \textbf{399}, 157 (2005).


\end{thebibliography}
\end{document}